\documentclass[conference]{IEEEtran}

\IEEEoverridecommandlockouts

\usepackage{cite}
\usepackage{bm}

\ifCLASSINFOpdf
  \usepackage[pdftex]{graphicx}
  \graphicspath{{./Images/}}
  \DeclareGraphicsExtensions{.pdf,.jpeg,.png}
\else
\fi

\usepackage[cmex10]{amsmath}
\usepackage{amsthm}

\allowdisplaybreaks
\usepackage{url}
\usepackage[hidelinks,draft=false]{hyperref}
\usepackage{amssymb}
\usepackage{multirow}
\usepackage{booktabs,colortbl}
\usepackage{xcolor}
\usepackage{siunitx}
\usepackage{psfrag}
\usepackage{commath}
\usepackage{flushend} 
\usepackage[nonumberlist,xindy]{glossaries}
\usepackage{enumitem}
\usepackage{mathtools}
\def\BibTeX{{\rm B\kern-.05em{\sc i\kern-.025em b}\kern-.08em
    T\kern-.1667em\lower.7ex\hbox{E}\kern-.125emX}}
\usepackage[margin=0.5in]{geometry}    
    
\newacronym{opf}{OPF}{optimal power flow}
\newacronym{res}{RES}{renewable energy sources}
\newacronym{eds}{EDS}{electrical distribution systems}
\newacronym{milp}{MILP}{mixed-integer linear programming}
\newacronym{minlp}{MINLP}{mixed-integer nonlinear programming}
\newacronym{misocp}{MISOCP}{mixed-integer second-order cone programming}
\newacronym{dg}{DG}{distributed generator}
\newacronym{der}{DER}{distributed energy resource}
\newacronym{ess}{ESS}{Energy storage systems}
\newacronym{mcs}{MCS}{Monte Carlo simulation}
\newacronym{wt}{WT}{wind turbine}
\newacronym{pv}{PV}{Photovoltaic}
\newacronym{soc}{SOC}{state of charge}
\newacronym{pdf}{PDF}{probability density function}
\newacronym{PDF}{PDF}{Probability density functions}
\newacronym{cdf}{CDF}{cumulative distribution function}
\newacronym{oem}{OEM}{optimal energy management}
\newacronym{vsc}{VSC}{voltage-source converter}
\newacronym{csc}{CSC}{current-source converter}
\newacronym{vpl}{VPL}{virtual power line}

\setglossarystyle{super}

\definecolor{yaxis}{rgb}{0.85, 0.33, 0.10}

\usepackage{textcomp}

\usepackage{url}

\usepackage{cite}

\usepackage{geometry}   
\newtheorem{prop}{Proposition}

\begin{document}

\title{{A Compensation Mechanism for EV Flexibility Services using Discrete Utility Functions} 
\thanks{This work was supported by the Netherlands Enterprise Agency (RVO) DEI+ project~{120037} ``\textit{Het Indië terrein: Een slimme buurtbatterij in de oude weverij}".}
}

\author{\IEEEauthorblockN{Juan S. Giraldo\IEEEauthorrefmark{1}, Nataly Ba\~{n}ol Arias\IEEEauthorrefmark{1}, Edgar Mauricio Salazar Duque\IEEEauthorrefmark{2}, Gerwin Hoogsteen\IEEEauthorrefmark{1},\\ and Johann L. Hurink\IEEEauthorrefmark{1}}
\IEEEauthorrefmark{1} dept. of Electrical Engineering, Mathematics and Computer Science, University of Twente, the Netherlands\\
\IEEEauthorrefmark{2} Electrical Energy Systems group, Eindhoven University of Technology, the Netherlands
}

\maketitle

\begin{abstract}
Compensation mechanisms are used to counterbalance the discomfort suffered by users due to quality service issues. Such mechanisms are currently used for different purposes in the electrical power and energy sector, e.g., power quality and reliability. This paper proposes a compensation mechanism using EV flexibility management of a set of charging sessions managed by a charging point operator (CPO). Users' preferences and bilateral agreements with the CPO are modelled via discrete utility functions for the energy not served. A mathematical proof of the proposed compensation mechanism is given and applied to a test scenario using historical data from an office building with a parking lot in the Netherlands. Synthetic data for 400 charging sessions was generated using multivariate elliptical copulas to capture the complex dependency structures in EV charging data. Numerical results validate the usefulness of the proposed compensation mechanism as an attractive measure both for the CPO and the users in case of energy not served.  

\end{abstract}


%


\section{Introduction}

%
%
%

The worldwide adoption of electric vehicles is rapidly increasing and electric mobility is expected to play a crucial role in future energy systems.
%
For example, in the Netherlands, it is expected that all new passenger vehicles sold must be zero-emission by 2030~\cite{rvo2022}. EVs lead to a considerable new demand volume, which, combined with the rapid pace of the EV adoption rates, might cause grid congestion and technical issues~\cite{limmer2019peak}. For instance, simultaneously charging a large number of EVs in a specific area may cause power outages and power quality issues~\cite{Nijenhuis_2021}.
Nonetheless, EVs also represent new potential flexibility sources due to their inherent load characteristics (manageable), which may help deal with technical issues in the distribution grid in the form of services~\cite{arias2019distribution}. 
Therefore, compensation mechanisms that allow EV users to provide flexibility services to charging point operators (CPOs) acting as aggregators, distribution system operators, or energy suppliers might be needed in future scenarios~\cite{Saner2022, Tsaousoglou_2021, tsaousoglou2022market}.
On the other hand, the emergence of quality policies for EV charging services in terms of energy supplied can be designed using compensation mechanisms. Such mechanisms might help ensure service quality by supporting customers' rights. 
%







Compensation mechanisms are already used for different aspects of power systems. For example, power supply reliability is an accepted aspect for which customers can be compensated after experiencing electricity outages, such as a power blackout~\cite{eur_com_2019}. According to the societal appreciation analysis in~\cite{eur_com_2019}, users are willing to pay (WTP) to avoid power outages but are also willing to accept (WTA) power outages under a proper monetary compensation, e.g., in the Netherlands, the WTP has been estimated to 0.61~€/kWh, and the WTA is around 25~€/kWh.  
Compensation mechanisms regarding power quality issues have mainly been used for commercial and industrial users. According to~\cite{Beleiu2018}, power quality disturbances such as voltage sag/swells and harmonic distortion cause direct and indirect costs due to interruptions, process slowdowns, and equipment failure, among others. These costs range from a few thousand to millions of euros, and network operators are obliged to compensate their customers~\cite{Beleiu2018}. 
%
%
Consequently, paying for complaints of power outages to commercial or residential users might be more expensive than compensating a group of EV users for having energy not served.

Concerning EV flexibility, corresponding services have been widely studied in the literature.
Authors in \cite{gadea2018market} proposed a compensation scheme based on a price model that considers reductions in the distribution transport tariff and the energy taxes to encourage EV grid support for DSOs while avoiding grid reinforcements.
~In \cite{Sarker2017}, EV owners receive monetary compensation for handing over the charging control of their vehicles to an aggregator to deal with congestion issues in the transformer. 

Monetary compensation for EV users providing ancillary services has also been studied in the literature. Authors in~\cite{Sun2021} and \cite{Yanchong2022} propose that users can get an extra economic compensation for participating in voltage regulation and in regulation markets, respectively. Similarly, three different tariff structures are proposed in~\cite{FREITASGOMES2021} to compensate EV owners for energy services, in which the demand is shifted toward a specific time interval. 
Other forms of compensations to EV owners for services provision are those related to battery degradation. For instance, a compensation mechanism based on carbon trading is proposed in~\cite{Qingyun2022} to reduce battery degradation costs.

EV users perceive any amount of energy not served as discomfort. However, depending on the user's preferences and needs, the level of discomfort might differ due to the energy not served. In this sense, a compensation mechanism for energy not served could be beneficial from several perspectives. For instance, the grid operator can avoid extra costs for grid expansion/reinforcements through flexibility procurement~\cite{brinkel2020should}, the CPO/aggregator can offer flexibility to the grid operator to improve its local objectives, and the EV owners' discomfort can be levelled in exchange for monetary compensation~\cite{Ricardo2022}. Furthermore, prosumers' utility functions can be nonlinear and nonconvex, or linear and flat with few large variations~\cite{Charalampos_2021}, motivating the use of discrete utility functions. 

This paper proposes a compensation mechanism for energy not served in EV charging sessions. Hereby, the flexibility services of EV users are taken into account via discrete utility functions mapping the energy not served. The contributions of this paper are as follows: 

\begin{itemize}
    \item A novel compensation mechanism that allows EV users to receive compensation for flexibility services via energy not served using discrete utility functions. 
    \item A general, flexible formulation for the utility functions of EV users that can be used in different power system applications regarding of its nature (convex or non-convex).  
    \item Application of the multivariate t-distribution elliptical copulas (MVT) for modeling residential load profiles given in~\cite{Duque_2021} to simulate synthetic EV charging sessions using real historical data.
\end{itemize}

The reminder of the paper is as follows: Section~\ref{sec_2} presents a contextualization of the problem and introduces the proposed compensation mechanism. Section~\ref{sec_3} contains the mathematical optimization model considering the discrete utility functions. The study case, tests, and results are presented in Section~\ref{sec_4}. Finally, the conclusions of the paper can be found in Section~\ref{sec_5}.   

\section{Problem Definition and Proposed Compensation Mechanism}  \label{sec_2}

\subsection{Problem definition}

We consider an aggregator or charging point operator (CPO) which is responsible for supplying a set of charging sessions \mbox{$\Omega_{\mathrm{N}} = \{1, 2, ...,  N\}$}, with $N$ being the number of charging sessions or EV users participating in the flexibility scheme over a scheduling horizon~$\Omega_{\mathrm{T}}$. We assume that each charging session has an energy requirement from the CPO under certain commercial agreements. Consider now a charging session \mbox{$n \in \Omega_{\mathrm{N}}$} which is characterized by the tuple $\mathrm{A}_n = (\mathrm{a}_n,\,\mathrm{d}_n,\,\mathrm{E}_n,\,\boldsymbol{u}_{n}, \,\gamma_{n})$, comprised by the session's arrival time $\mathrm{a}_n$, departure time $\mathrm{d}_n$, required charging energy $\mathrm{E}_n$, its utility function $\boldsymbol{u}_{n}$, and an acceptable energy percentage $\gamma_{n}$.

Notably, we assume that the utility function \mbox{($\boldsymbol{u}_{n}=u\left(\phi_{n}\right)$)} represents a bilateral agreement between the CPO and the charging session as a function of the energy not served at the session's departure time ($\phi_{n}$). This agreement implies that a session $n$ complies with a lower final state of charge than initially requested; however, only in exchange for compensation. Furthermore, the CPO gains \textit{certainty} on the sessions' characteristics and controllability of its state. It is noteworthy that achieving this certainty requires that the user shares and satisfies the promises in $\mathrm{A}_n$, which can be tackled using truthful compensation and penalty mechanisms as in~\cite{Tsaousoglou_2021}. Hence, in this paper we assume that $\mathrm{A}_n$ is a truthful declaration from each charging session.


\subsection{Proposed compensation mechanism}

Let $\mathrm{c}_{n,t}$ denote the dynamic electricity price session $n$ must pay at time $t$, $x_{n,t}$ the power consumption allocated to the session at that period, and $\Delta\mathrm{t}$ the time length of the discretized scheduling horizon. Then, the total energy cost a user must pay to the CPO for the served energy is given by
\begin{flalign}\label{energy_cost_user}
& \pi_{n}\left(x_{n,t},\phi_{n}\right)\!=\!\Delta\mathrm{t}\sum_{t\in\Omega_{\mathrm{T}}}\mathrm{c}_{n,t}x_{n,t} - {u}\left(\phi_{n}\right) & \forall n\in \Omega_{\mathrm{N}}
\end{flalign}
\noindent where ${u}\left(\phi_{n}\right)$ represents a non-negative function mapping the energy not served. For the sake of simplicity, we assume that utility functions are rational in the game-theoretic sense, i.e., they are monotonic increasing with the energy not served; then,~$\max\{u\left(\phi_{n}\right)\}\xrightarrow[]{}\phi_{n}=\mathrm{E}_{n}$.

It is of the CPO's interest to pursue that~\eqref{energy_cost_user} is \mbox{\textit{revenue adequate}}, i.e., when the total energy cost received from all participants is greater than or equal to zero~\cite{Ming_2018}. Guaranteeing revenue adequacy ensures the CPO to avoid budget deficit. Let \mbox{$\mathrm{e}_{n}$} be the acceptable fraction of the energy required by session $n$ ($\mathrm{e}_{n}=\gamma_{n}\mathrm{E}_{n}$), and \mbox{$\mathrm{C}_{n}=\min\limits_{t\in\Omega_{\mathrm{T}}}\{\mathrm{c}_{n,t}\}$} be the minimum electricity price charged to the session during the scheduling horizon. Similarly, the supplied energy to a session is expressed as \mbox{$\mathcal{H}_{n}=\Delta\mathrm{t}\sum\limits_{t\in\Omega_{\mathrm{T}}}x_{n,t}$}. Then, the following is stated:
\begin{prop}\label{prop_1}
If $\boldsymbol{u}_{n}$ is capped such that \mbox{$u\left(\phi_{n}\right)\leq\mathrm{C}_{n}\mathrm{e}_{n}$} for each charging session $n\in\Omega_{N}$, the energy cost in~\eqref{energy_cost_user} is revenue adequate for the CPO if at least the acceptable energy is served to all sessions.  
\end{prop}
\begin{proof}
The payment for the energy served by the CPO from session $n$ is bounded from below by $\mathrm{C}_{n}\mathcal{H}_{n}$. Notice that the energy not served to the session is defined as \mbox{$\phi_{n} = \mathrm{E}_{n}-\eta_{n}\mathcal{H}_{n}$}, where $\eta_{n}$ is the session's charging efficiency. The efficiency will be treated here as $\eta_{n}=1$ for simplicity.  
Assume now the case where exactly the acceptable energy fraction is served, i.e., \mbox{$\mathcal{H}_{n}=\mathrm{e}_{n}$}, implying that \mbox{$\phi_{n}^{\prime}=\mathrm{E}_{n}-\mathrm{e}_{n}=\left(1-\gamma_{n}\right)\mathrm{E}_{n}$}, and \mbox{$u\left(\phi_{n}^{\prime}\right)\leq u\left(\mathrm{E}_{n}\right)\leq \mathrm{C}_{n}\mathrm{e}_{n}$}. Thus, the payment for the acceptable energy is revenue adequate since $\pi_{n}\geq0$, where the total energy cost is $\pi_{n}=\mathrm{C}_{n}\mathcal{H}_{n} - u\left(\phi_{n}^{\prime}\right)$; which implies that $\pi_{n}< 0$ is conditioned to~\mbox{$\mathcal{H}_{n}<\mathrm{e}_{n}$}.

Hence, the energy cost is always revenue adequate for the CPO ($\pi_{n}>0$) as long as the CPO supplies at least the acceptable energy to all sessions and \mbox{$u\left(\phi_{n}\right)\leq\mathrm{C}_{n}\mathrm{e}_{n}$}.  
\end{proof}

As a remark, we point out that a reduction in the total final energy cost ($\mathcal{P}={\sum\limits_{n} \pi_{n}}$) would not necessarily reflect a drop in the revenue of the CPO since the latter may be able to receive a settlement from external parties due to flexibility services, e.g., by its participation in a local flexibility market~\cite{tsaousoglou2022market}.

\section{Mathematical Model for EV Flexibility Management and Discrete Utility Functions} \label{sec_3}


\subsection{Charging Sessions}
We assume a reservation system is given that accepts charging sessions for the planning horizon, for example, in a day-ahead scheme. The power consumption $x_{n,t}$ is a decision variable determined by the CPO. Based on the sessions' characteristics~$\mathrm{A}_n$, the following constraints are formulated: No power can be allocated to the session~$n$ before its arrival time or after its departure time:
\begin{flalign} \label{arrival}
&x_{n,t}=0,&\forall\, n \in \Omega_{\mathrm{N}},\, t \in \Omega_{\mathrm{T}}\,:\,t<\mathrm{a}_n \lor t>\mathrm{d}_{n}
\end{flalign}
%
%
\noindent while the power allocation of each session is bounded by the charging pole characteristics in which a charging session is connected:
\begin{flalign} \label{bounds}
& 0 \leq x_{n,t} \leq \overline{\mathrm{x}_n}, &  \forall\, n \in \Omega_{\mathrm{N}},\, t \in \Omega_{\mathrm{T}}
\end{flalign}

Furthermore, all energy requirements should be \textit{ideally} satisfied during the charging session:
\begin{flalign} \label{serve}
&\mathcal{H}_{n}=\Delta\mathrm{t}\sum\limits_{t\in\Omega_{\mathrm{T}}}x_{n,t}, & \forall\, n \in \Omega_{\mathrm{N}}\\
&\mathrm{E}_n \leq \phi_{n} +  \eta_{n}\mathcal{H}_{n} , & \forall\, n \in \Omega_{\mathrm{N}}.
\end{flalign}
%

\subsection{Users' Utility Functions}

Notice that the energy not served represents discomfort to the user. Hence, we propose the use of discrete utility functions to represent bilateral agreements between the CPO and the charging sessions, as depicted in Fig.~\ref{util_disc}.

Formally, we use a lower-semicontinuous piecewise-linear function representing the monotonic increasing utility function of session $n$ as:
\begin{flalign} \label{util_funct_original}
\boldsymbol{u}_{n} = \begin{cases}
      f_{n,0}\!=\!0 & \phi_{n}\!=\!0 \\
      \vdots \\ 
            f_{n,k}\!=\!\mathrm{h}_{n,k}\phi_{n}\!+\!\mathrm{b}_{n,k} & \alpha_{n,k-1}<\phi_{n} \leq \alpha_{n,k} \\
      \vdots \\
      f_{n,\kappa}\!=\!\mathrm{h}_{n,\kappa}\phi_{n} + \mathrm{b}_{n,\kappa} & \alpha_{n,\kappa-1} < \phi_{n} \leq \alpha_{n,\kappa}
   \end{cases}
\end{flalign}
\noindent  where $k\in\Omega_{\mathrm{K}}^{n}=\left\{0,\,...,\,\kappa\right\}$ represents the set of break points for charging session $n$, while $\{\mathrm{h}_{n,k},\,\mathrm{b}_{n,k}\}$, and $\{\alpha_{n,k-1},\,\alpha_{n,k}\}$ denote the coefficients and lower and upper bounds of the functions, respectively, with $k\geq1$. For the sake of simplicity, we define \mbox{$\overline{\mathrm{u}}_{n,k-1}:= f_{n,k}\left(\alpha_{n,k-1}\right)$} and \mbox{$\underline{\mathrm{u}}_{n,k}:= f_{n,k}\left(\alpha_{n,k}\right)$} as the value of the function for segment \mbox{$k$} evaluated on its endpoints. Note that~\eqref{util_funct_original} implies that \mbox{$\underline{\mathrm{u}}_{n,0}=\overline{\mathrm{u}}_{n,0}=\alpha_{n,0}=0$},  $\underline{\mathrm{u}}_{n,\kappa}\leq\mathrm{C}_{n}\mathrm{e}_{n}$, and $\alpha_{n,\kappa}=\mathrm{E}_{n}$ according to Proposition~\ref{prop_1}. Thus, by defining ${\mathcal{Z}}_{n}:=\boldsymbol{u}_{n}$ as an auxiliary variable representing the cost for the energy not served, and multipliers $\overline{\lambda}_{n,k}$, $\underline{\lambda}_{n,k}\; \forall\,k\in\Omega_{\mathrm{K}}^{n}$ as the weights on these two endpoints~\cite{croxton2003comparison}, the utility function can be expressed as a linear combination of the cost at the endpoints by:
\begin{flalign}\label{util_funct_piece}
&\mathcal{Z}_{n}\!=\!\sum_{\mathclap{{k\in\Omega_{\mathrm{K}}^{n}}:\,k<\kappa}}\left(\underline{\lambda}_{n,k}\underline{\mathrm{u}}_{n,k}\!+\!\overline{\lambda}_{n,k}\overline{\mathrm{u}}_{n,k}\right)\!+\!\underline{\lambda}_{n,\kappa}\underline{\mathrm{u}}_{n,\kappa} & \forall n \in \Omega_{\mathrm{N}}
\end{flalign}
\noindent where the energy not supplied is defined as:
\begin{flalign}
&\phi_{n}\!=\!\sum_{\mathclap{{k\in\Omega_{\mathrm{K}}^{n}}:\,k<\kappa}}\left(\underline{\lambda}_{n,k}\!+\!\overline{\lambda}_{n,k}\right)\alpha_{n,k}+\!\underline{\lambda}_{n,\kappa}\alpha_{n,\kappa} &\forall \, n\in\Omega_{\mathrm{N}}
\end{flalign}
\begin{flalign}
&1=\sum_{\mathclap{{k\in\Omega_{\mathrm{K}}^{n}}:\,k<\kappa}}\left(\underline{\lambda}_{n,k}\!+\!\overline{\lambda}_{n,k}\right) +\underline{\lambda}_{n,\kappa} &\forall\,s \in \Omega_{\mathrm{S}},\, n\in\Omega_{\mathrm{N}}
\end{flalign}
%
%
%
\begin{flalign}\label{utili1}
&\overline{\lambda}_{n,k}\!+\!\underline{\lambda}_{n,k+1}\!=\!y_{n,k+1} & \forall \, n\in\Omega_{\mathrm{N}}, k\in\Omega_{\mathrm{K}}^{n}:\,k<\kappa
\end{flalign}
\begin{flalign}\label{utili2}
&\sum_{{{k\in\Omega_{\mathrm{K}}^{n}}:k\geq1}} y_{n,k} \leq 1 & \forall \,n\in\Omega_{\mathrm{N}}
\end{flalign}
\begin{flalign}
\label{binary_var}
&\overline{\lambda}_{n,k},\, \underline{\lambda}_{n,k} \geq 0,\; y_{n,k}\in\{0,1\} &\forall \, n\in\Omega_{\mathrm{N}},\, k\in\Omega_{\mathrm{K}}^{n}
\end{flalign}

\begin{figure}[!t]
\centering
  \psfrag{u}[][][1]{{$\boldsymbol{u_{n}}$}}
   \psfrag{a0}[][][1]{{$\alpha_{n,0}$}}
 \psfrag{a1}[][][1]{{$\alpha_{n,1}$}}
  \psfrag{a2}[][][1]{{$\alpha_{n,\kappa-1}$}}
  \psfrag{a3}[][][1]{{$\alpha_{n,\kappa}$}}
  \psfrag{phi}[][][1]{{$\boldsymbol{\phi_{n}}$}}
  \psfrag{f}[][][0.8]{{$f_{n,\kappa-1}$}}
    \psfrag{f1}[][][0.8]{{$f_{n,1}$}}
  \psfrag{f3}[][][0.8]{{$f_{n,\kappa}$}}
  \psfrag{uc}[][][0.8]{{$\mathrm{C}_{n}\mathrm{e}_{n}$}}
  \psfrag{b1}[][][1]{{$\overline{\mathrm{u}}_{n,1}$}}
  \psfrag{b2}[][][1]{{$\underline{\mathrm{u}}_{n,\kappa-1}$}}
\includegraphics[width=0.9\columnwidth]{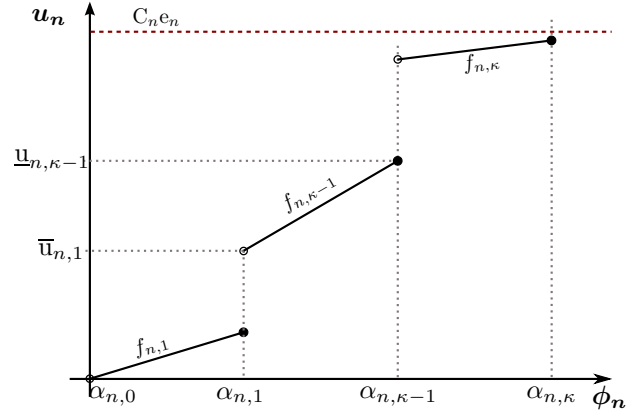}
\caption{Representation of an arbitrary discrete utility function at a charging session $n$ with $\kappa=3$.}
\label{util_disc}
\end{figure}

Notice that~\eqref{utili1} ensures that the multipliers are only different from zero in the segment where $y_{n,k}$ is activated, while~\eqref{utili2} guarantees that at the most one segment can be active. Summarizing, the set of variables from the charging sessions is defined as~\mbox{$\mathcal{Y}_{cs}=\{\mathcal{Z}_{n}, \overline{\lambda}_{n,k},\underline{\lambda}_{n,k}, y_{n,k}\}$}. Moreover, the discrete utility function in~\eqref{util_funct_original} can be used to approximate any given nonlinear, nonconvex, or linear and flat function with few large variations, making it suitable to represent prosumers~\cite{Charalampos_2021}.

\subsection{Charging Point Operator (CPO)}
The power consumption profile of the CPO through the scheduling horizon~$\Omega_{\mathrm{T}}$, is given by:
\begin{flalign} \label{p_station}
& p_{t} = \sum_{n \in \Omega_{\mathrm{N}}}x_{n,t}, &\forall\, t \in \Omega_{\mathrm{T}}   
\end{flalign}

\noindent which is constrained by upper/lower bounds representing the power capacity of the station's transformer or a flexibility request from the main grid $\overline{\mathrm{p}}$, and an energy limit $\overline{{\mathcal{E}}}$ representing eventual energy contracts with the energy supplier as follows:
\begin{flalign}
&0 \leq p_{t} \leq \overline{\mathrm{p}}, & \forall \, t \in \Omega_{\mathrm{T}}\\
\label{ener_max}
&\Delta{\mathrm{t}} \sum_{t \in \Omega_{\mathrm{T}}}p_{t}\leq \overline{{\mathcal{E}}} &
\end{flalign}

The aim of the CPO is to minimize the cost for energy not served to the charging sessions while considering its own operational limits. Furthermore, it is the CPO's responsibility to guarantee the operational conditions of the charging pool while aiming for an economically optimal operation.  Hence, the set of decision variables of the CPO is defined as $\mathcal{Y}_{cpo}=\{x_{n,t},\,p_{t}\}$. The optimal operation of the CPO considering the proposed compensation mechanism is given by:
\begin{flalign}
\begin{split}
\label{opf_nonconvex}
&\max_{\boldsymbol{\mathcal{Y}}}\hspace{5pt}  \Delta\mathrm{t}\sum_{t\in\Omega_{\mathrm{T}}}\sum_{n\in\Omega_{\mathrm{N}}}\mathrm{c}_{n,t}x_{n,t}  - \sum_{n\in\Omega_{\mathrm{N}}}\mathcal{Z}_{n}\\
&\quad \mathrm{\mathrm{s.t.}}
\hspace{20pt} \eqref{arrival}\textrm{--}\eqref{serve} ,\, \eqref{util_funct_piece}\textrm{--}\eqref{ener_max}
\end{split}
\end{flalign} 
\noindent where the set \mbox{$\boldsymbol{\mathcal{Y}}=\{\mathcal{Y}_{cs} \cup \mathcal{Y}_{cpo}\}$} determines the decision variables of the model.

\section{Study Case, Tests, and Results}  \label{sec_4}

\begin{figure}[!t]
\centering
\includegraphics[width=\columnwidth]{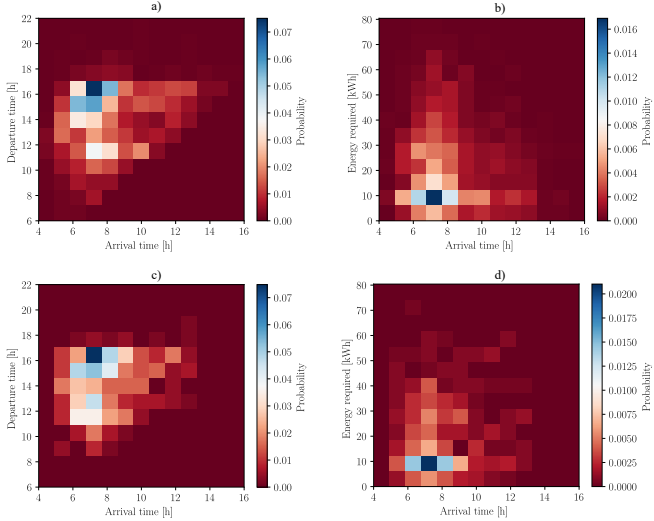}
\caption{Distribution of charging session profiles: (a) and (b) measured historical data. (c) and (d) synthetically generated data.}
\label{used_data}
\end{figure}

The proposed compensation mechanism and optimal operation CPO model are tested using historical data from an office building with charging stations in the Netherlands. The database consists of a total of $1{,}760$ charging sessions from the $1^{\mathrm{st}}$~of~June to $31^{\mathrm{st}}$~of~December 2021, containing arrival time, departure time, and energy required. The distribution of the measured data relating the arrival time and the departure time is displayed in Fig.~\ref{used_data}~(a), while the distribution relating the arrival time and the energy required is depicted in Fig.~\ref{used_data}~(b). The procedure in~\cite{Duque_2021} is followed to find the marginal distributions representing the arrival times, departure times, and required energy. With these distributions, $400$ sessions were synthetically generated, resulting in the distributions displayed in Fig~\ref{used_data}~(c) and (d). It can be seen that the obtained synthetic sessions follow the original's data distribution, highlighting the potential of MVT copulas for modeling EV charging sessions~\cite{einolander2022multivariate}.

A maximum power of $\overline{\mathrm{x}}_{n}=22\,$kW has been set for all poles, the charging efficiency $\eta_{n}=0.92$, while the transformer rating is set to $\overline{\mathrm{p}} = 700\,$kVA and the energy limit as $\overline{\mathcal{E}} = 8.2\,$MWh, unless explicitly stated. To simulate diversity among sessions, the acceptable energy percentage for each session is assume to follow a uniform distribution, \mbox{$\gamma_{n}\sim \mathrm{U}\left(0.5,1.0\right)$}. Similarly, all utility functions have been chosen randomly with $\kappa=4$ break points and the compensation capped according to Proposition~\ref{prop_1} (\mbox{$\underline{\mathrm{u}}_{n,\kappa}\leq \mathrm{C}_{n}\mathrm{e}_{n}$}). Furthermore, two fixed electricity rates have been used in this case study:
\begin{equation}
    \mathrm{C}_{n} = \begin{cases}
                    0.30\,\mathrm{\text{\texteuro}/kWh}, & \mathrm{if} \; \gamma_{n} < 0.99 \\
                    0.35\,\mathrm{\text{\texteuro}/kWh}, & \mathrm{otherwise}
                     \end{cases}
                     \quad \forall n \in \Omega_{\mathrm{N}}
\end{equation}

\noindent with $\mathrm{c}_{n,t}=\mathrm{C}_{n}\,\forall\,t\in\Omega_{\mathrm{T}}$, representing charging sessions participating in the compensation mechanism and others not willing to participate, respectively. Two main tests have been performed to assess the impact of reducing CPO's limits on the user compensation and total energy cost:

\subsection{T1 - The CPO is able to serve at least the acceptable energy to all sessions}

The total power and cumulative energy considering the base case limits is depicted in Fig.~\ref{P_and_E}~(a), where the sessions' total required energy considering the charging efficiency, $\sum\limits_{n} \mathrm{E}_{n}/\eta_{n} = 8.1\,$MWh is plotted. It can be seen that the CPO is able to supply the required energy, indicating that there is no energy not served in the base case.

\begin{figure}[!t]
\centering
\includegraphics[width=\columnwidth]{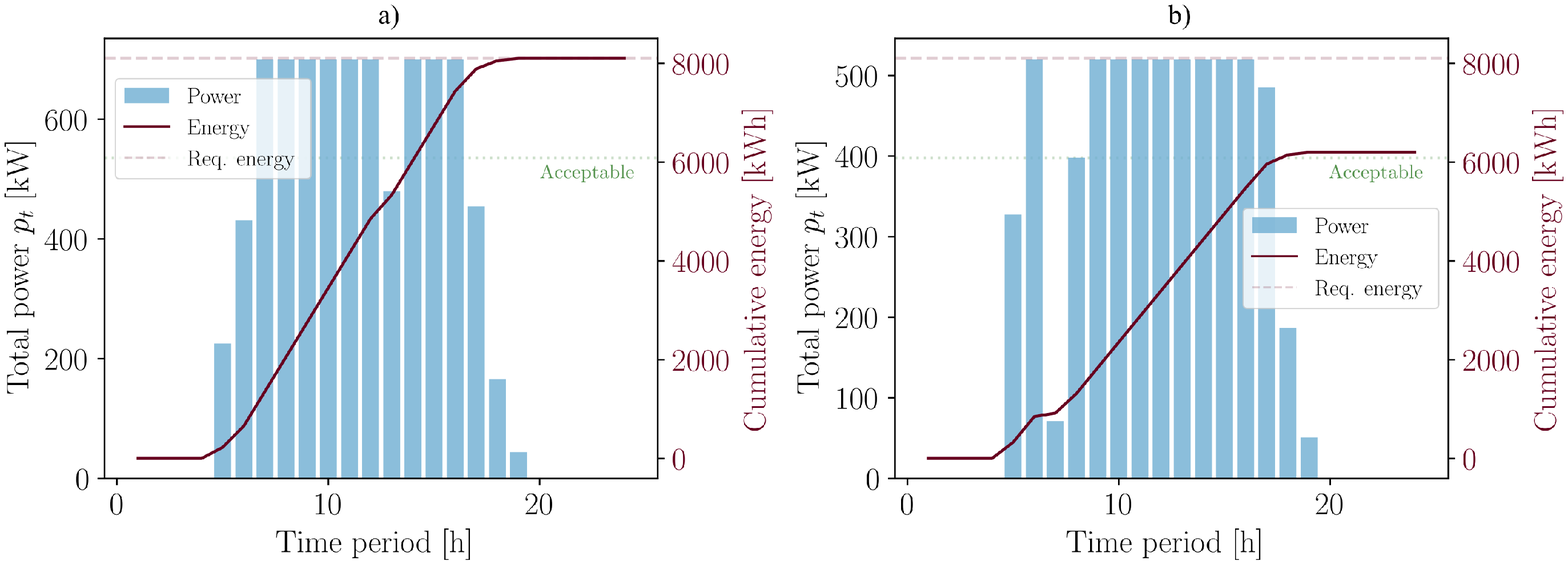}
\caption{CPO's power consumption and cumulative energy. \mbox{(a)~Energy limit~$\overline{\mathcal{E}}=8.2\,$MWh}. (b)~Energy limit~$\overline{\mathcal{E}}=6.2\,$MWh.}
\label{P_and_E}
\end{figure}

The CPO's limits are now reduced to $\overline{\mathrm{p}} = 520\,$kVA and the energy limit to \mbox{$\overline{\mathcal{E}} = 6.2\,$MWh} this may resemble an agreement with the energy supplier, a peak shave procedure, or a response to a flexibility procurement from the main grid. The resulting total power and the cumulative energy is shown in~Fig.~\ref{P_and_E}(b), where it can be seen that the CPO is able to serve more than the acceptable energy, \mbox{$\sum\limits_{n} \mathrm{e}_{n}/\eta_{n} = 6.065\,$MWh}. As expected, reducing the operational limits modifies the power profile during the planning horizon while the CPO tries to maximize its revenue. This modification leads to $6.2$\,MWh of served energy and approximately $1.9$\,MWh of energy not served, which is distributed among the charging sessions. The probability density function (PDF) and empirical cumulative distribution function (CDF) of the energy served and energy not served are depicted in Fig.~\ref{cdf_pdf} (a) and (b), respectively. From Fig.~\ref{cdf_pdf}~(a), it can be seen that $70\%$ of the sessions received at least $20$\,kWh, whereas from Fig.~\ref{cdf_pdf}~(b) can be seen that $90\%$ of the sessions had at least $10$\,kWh of energy not served.        

\begin{figure}[!b]
\centering
\includegraphics[width=\columnwidth]{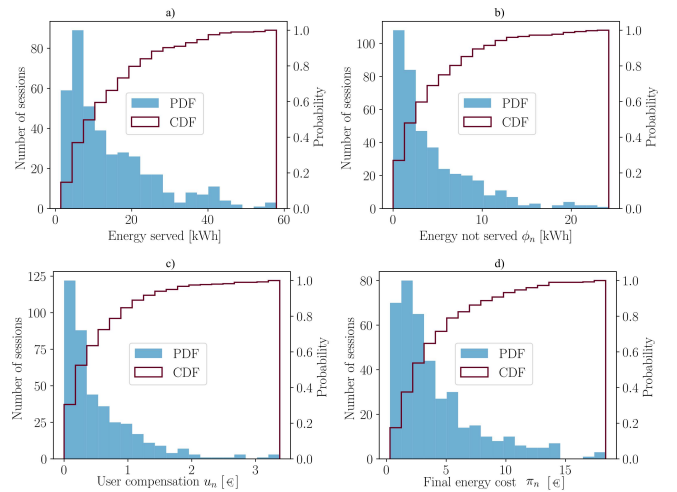}
\caption{T1 - PDF and CDF for the charging sessions. (a) Energy served. (b)~Energy not served. (c)~User compensation. (d)~Final energy cost.}
\label{cdf_pdf}
\end{figure}

The distribution of the energy not served is based on the cost for compensation defined by the users' utility functions. The PDF and CDF of the user compensation is depicted in Fig.~\ref{cdf_pdf}~(c), where it can be seen that $90\%$ of the sessions receive a compensation lower than \texteuro~$1.30$. Similarly, Fig.~\ref{cdf_pdf}~(d) shows the PDF and CDF of the final energy cost of all charging sessions, indicating that $80\%$ of the sessions pay less than \texteuro~$6.00$ for their final energy cost. This reduction in the energy cost translates into an equivalent electricity rate of \textcent/kWh\,$26.59$ in average, and a reduction of around $32\%$ in the total final energy cost of the~CPO. It is noteworthy in Fig.~\ref{cdf_pdf}~(d) that the final energy cost of all sessions is non-negative ($\pi_{n}>\textrm{\texteuro}\,0.29$), indicating the solution is revenue adequate for all sessions.

\begin{figure}[!t]
\centering
\includegraphics[width=\columnwidth]{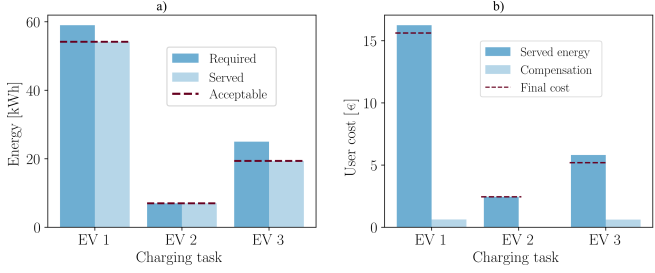}
\caption{T1 - Selected charging sessions. (a) Energy served, required, and acceptable. (b) Users' charging cost, compensation, and total.}
\label{sel_EVs}
\end{figure}

Three charging sessions (EV~1, EV~2, and EV~3) have been arbitrarily selected to further analyze the effect of the proposed compensation mechanism. Particularly, EV~1 represents a charging session with high compensation costs, required energy $\mathrm{E}_{1}=59$\,kWh and acceptable energy $\gamma_{1} = 0.918$; EV~2 represents a charging session not willing to participate in the compensation mechanism with $\mathrm{E}_{2}=7$\,kWh; and EV~3 represents a session with low compensation cost, $\mathrm{E}_{3}=25$\,kWh and $\gamma_{3} = 0.775$. The required energy, served energy, and acceptable energy of these three sessions is depicted in Fig.~\ref{sel_EVs}~(a). It can be seen that at least the acceptable energy is served to all sessions, indicating an energy not served of $4.84$\,kWh for EV~1 and of $5.63$\,kWh  for EV~3. However, as can be seen in Fig.~\ref{sel_EVs}~(b), the final energy cost is non-negative in all sessions with a compensation of \texteuro\,$0.63$ for EV~1 and \texteuro\,$0.62$ for EV~3 and a final cost of \texteuro\,$17.03$ for EV~1, \texteuro\,$2.66$ for EV~2, and \texteuro\,$5.69$ for EV~3.

\subsection{T2 - The CPO is unable to serve at least the acceptable energy to all sessions}

In this case, the CPO's limits are reduced to \mbox{$\overline{\mathrm{p}}= 520\,$}kVA and the energy limit to \mbox{$\overline{\mathcal{E}} = 5.0\,$MWh}, resembling a case where the CPO is unable to serve all the acceptable energy of the sesions. This reduction leads to $5.0$\,MWh of served energy and approximately $3.1$\,MWh of energy not served. The PDF and empirical CDF of the energy served and not served are depicted in Fig.~\ref{cdf_pdf_5000} (a) and (b), respectively. From Fig.~\ref{cdf_pdf_5000}~(a), it can be seen that $70\%$ of the sessions received at least $12$\,kWh, which is $60\%$ of the $20$\,kWh in T1. On the other hand, from Fig.~\ref{cdf_pdf_5000}~(b) it can be seen that $90\%$ of the sessions had at least $13$\,kWh of energy not served, which is $30\%$ more energy not served than in T1.

The PDF and CDF of the user compensation is depicted in Fig.~\ref{cdf_pdf_5000}~(c), where it can be seen that $90\%$ of the sessions receive a compensation lower than \texteuro~$1.69$, which is $30\%$ higher than in T1. Similarly, Fig.~\ref{cdf_pdf_5000}~(d) shows the PDF and CDF of the final energy cost of all charging sessions, indicating that $80\%$ of the sessions pay less than \texteuro~$3.72$ for their final energy cost, representing a reduction of $38\%$ compared to T1. This reduction in the energy cost translates into an equivalent electricity rate of \textcent/kWh\,$21.47$ in average, and a reduction of around $56\%$ in the total final energy cost of the CPO. Note that approximately $30\%$ of the sessions have negative final energy costs (lower than or equal to \texteuro\,$-0.22$), indicating that the solution is not revenue adequate for all sessions.

\begin{figure}[!t]
\centering
\includegraphics[width=\columnwidth]{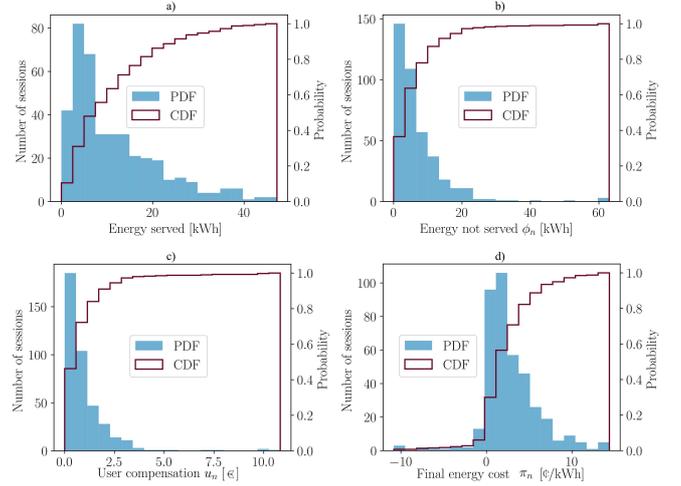}
\caption{T2 - PDF and CDF for the charging sessions. (a) Energy served. (b)~Energy not served. (c)~User compensation. (d)~Final energy cost.}
\label{cdf_pdf_5000}
\end{figure}

\begin{figure}[!t]
\centering
\includegraphics[width=\columnwidth]{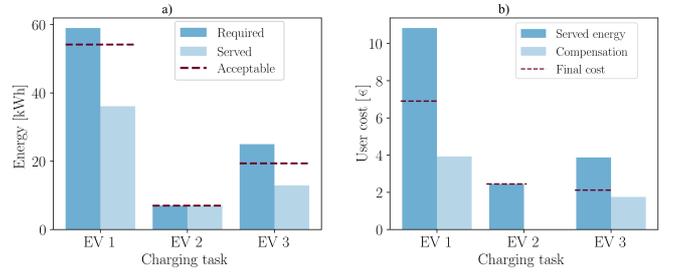}
\caption{T2 - Selected charging sessions. (a) Energy served, required, and acceptable. (b) Users' charging cost, compensation, and total.}
\label{sel_EVs_5000}
\end{figure}

In a next step, the same charging sessions as selected in T1 are studied. The required energy, served energy, and acceptable energy of these three sessions is depicted in Fig.~\ref{sel_EVs_5000}~(a), where it can be seen that EV~2 is the only session with at least the acceptable energy served. Notice that both EV~1 and EV~3 received less that the acceptable energy. However, as can be seen in Fig.~\ref{sel_EVs_5000}~(b), the final energy cost is non-negative in all sessions with a compensation of \texteuro\,$3.93$ for EV~1 and \texteuro\,$1.75$ for EV~3 and a final cost of \texteuro\,$7.84$ for EV~1, \texteuro\,$2.66$ for EV~2, and \texteuro\,$2.46$ for EV~3. This result is due to the noncontinuous characteristic of the utility functions.

Finally, Table~\ref{table_1} summarizes the sensitivity of the total cost of served energy, total user compensation ($\mathcal{U}={\sum\limits_{n} u_{n}}$), and total final energy cost ($\mathcal{P}={\sum\limits_{n} \pi_{n}}$) with respect to the CPO's power and energy capacity. Also, the minimum final energy cost of all sessions ($\rho = \min\limits_{n}\phi_{n}$), the total energy not served ($\Phi = \sum\limits_{n}\phi_{n}$), and the average unitary cost can be found in Table~\ref{table_1}. The first row in Table~\ref{table_1} represents the base case, where all the required energy is served to the sessions, \mbox{$\Phi = 0$}. The impact of reducing the peak power imported from the main grid is evident in the second row, where the energy not served increases although the energy limit is not reduced. As expected, due to the proposed compensation mechanism, the average unitary energy cost is inversely proportional to the energy not served. Notice in Table~\ref{table_1} that the compensation mechanism is revenue adequate as long as the CPO is able to supply at least the acceptable energy to the sessions ($6.065\,$MWh), which is evident as soon as $\rho$ becomes negative. On the other hand, also notice that the total revenue $\mathcal{P}$ is non-negative even even after a reduction of $63\%$ in the energy limit. The point where $\mathcal{P}$ could be interpreted as the limit where the CPO might no longer offer flexibility to the main grid operator.   

\begin{table}[!t]
\centering
\caption{Sensitivity of the compensation mechanism to the CPO's power and energy capacity.}
\label{table_1}
\resizebox{\columnwidth}{!}{%
\begin{tabular}{@{}cccccccc@{}}
\toprule
$\boldsymbol{\overline{\mathrm{p}}}$ & $\boldsymbol{\overline{\mathcal{E}}}$ & \textbf{Served cost} & $\boldsymbol{\mathcal{U}}$ & $\boldsymbol{\mathcal{P}}$ & $\boldsymbol{\rho}$ &  $\boldsymbol{\Phi}$ & \textbf{Average} \\
MW & MWh & \texteuro & \texteuro & \texteuro & \texteuro & MWh & \textcent/kWh \\ \midrule
0.7 & 8.2 & 2436.54 & 0.00 & 2436.54 & 0.65 & \textbf{0.00} & 30.08 \\
0.52 & 8.2 & 2066.40 & 130.82 & 1935.59 & 0.37 & \textbf{1.135} & 28.19 \\
0.52 & 6.2 & 1866.30 & 217.84 & 1648.46 & 0.29 & 1.748 & 26.59 \\
0.52 & 6.07 & 1827.30 & 235.78 & 1591.52 & \textbf{0.29} & 1.868 & 26.22 \\
0.52 & 6.06 & 1824.30 & 237.19 & 1587.12 & \textbf{-1.83} & 1.877 & 26.19 \\
0.52 & 5 & 1506.30 & 432.97 & 1073.34 & -10.83 & 2.852 & 21.47 \\
0.52 & 4 & 1206.30 & 650.05 & 556.25 & -11.81 & 3.772 & 13.91 \\
0.52 & 3 & 906.30 & 885.83 & \textbf{20.48} & -12.84 & 4.692 & 0.68 \\
0.52 & 2.95 & 891.30 & 897.93 & \textbf{-6.63} & -12.84 & 4.738 & -0.22 \\ \bottomrule
\end{tabular}
}
\end{table}

    

\section{Conclusions} \label{sec_5}

In this paper, a compensation mechanism using EV flexibility management via energy not served has been introduced. The preferences of users and their bilateral agreements with the CPO have been modeled using discrete utility functions for the energy not served. A proof for the proposed compensation mechanism considering the particular characteristics of the charging sessions has been proposed. The proposed compensation mechanism was applied to a test scenario based on historical data from an office building with a parking lot in the Netherlands. Using the historical data set, synthetic data for 400 charging sessions was generated using multivariate \mbox{t-distribution} elliptical copulas. This approach allows capturing the complex dependency structures in EV charging data. Results validate the use of the compensation mechanism as a revenue adequate measure for the CPO as long as it is able to supply at least the acceptable energy to the sessions. At the same time, EV users perceive a reduction in the charging electricity rate as compensation according to their utility function. Hence, the proposed compensation mechanism is an attractive measure both for the CPO and the users. Future work will analyze the proposed compensation mechanism taking into account multiple CPOs as flexibility aggregators, the main grid operator, and a flexibility market structure. 


\bibliographystyle{IEEEtran}
\bibliography{references}

\begin{thebibliography}{10}
\providecommand{\url}[1]{#1}
\csname url@samestyle\endcsname
\providecommand{\newblock}{\relax}
\providecommand{\bibinfo}[2]{#2}
\providecommand{\BIBentrySTDinterwordspacing}{\spaceskip=0pt\relax}
\providecommand{\BIBentryALTinterwordstretchfactor}{4}
\providecommand{\BIBentryALTinterwordspacing}{\spaceskip=\fontdimen2\font plus
\BIBentryALTinterwordstretchfactor\fontdimen3\font minus
  \fontdimen4\font\relax}
\providecommand{\BIBforeignlanguage}[2]{{%
\expandafter\ifx\csname l@#1\endcsname\relax
\typeout{** WARNING: IEEEtran.bst: No hyphenation pattern has been}%
\typeout{** loaded for the language `#1'. Using the pattern for}%
\typeout{** the default language instead.}%
\else
\language=\csname l@#1\endcsname
\fi
#2}}
\providecommand{\BIBdecl}{\relax}
\BIBdecl

\bibitem{rvo2022}
{Netherlands Enterprise Agency (RVO)}, ``Electric vehicles statistics in the
  {Netherlands},'' \url{https://bit.ly/3PsDwFI}, 2022, [Online: accessed May
  -2022].

\bibitem{limmer2019peak}
S.~Limmer and T.~Rodemann, ``Peak load reduction through dynamic pricing for
  electric vehicle charging,'' \emph{International Journal of Electrical Power
  \& Energy Systems}, vol. 113, pp. 117--128, 2019.

\bibitem{Nijenhuis_2021}
B.~Nijenhuis, M.~E. Gerards, G.~Hoogsteen, and J.~L. Hurink, ``An outage
  probability model for electric vehicles in low voltage grids,'' in \emph{2021
  IEEE PES Innovative Smart Grid Technologies Europe (ISGT Europe)}, {E}spoo,
  {F}inland, 18-21 Oct. 2021.

\bibitem{arias2019distribution}
N.~B. Arias, S.~Hashemi, P.~B. Andersen, C.~Tr{\ae}holt, and R.~Romero,
  ``Distribution system services provided by electric vehicles: recent status,
  challenges, and future prospects,'' \emph{IEEE Transactions on Intelligent
  Transportation Systems}, vol.~20, no.~12, pp. 4277--4296, Dec. 2019.

\bibitem{Saner2022}
C.~B. Saner, A.~Trivedi, and D.~Srinivasan, ``A cooperative hierarchical
  multi-agent system for {EV} charging scheduling in presence of multiple
  charging stations,'' \emph{IEEE Transactions on Smart Grid}, vol.~13, no.~3,
  pp. 2218--2233, May 2022.

\bibitem{Tsaousoglou_2021}
G.~Tsaousoglou, J.~S. Giraldo, P.~Pinson, and N.~G. Paterakis, ``Mechanism
  design for fair and efficient {DSO} flexibility markets,'' \emph{IEEE Trans.
  Smart Grid}, vol.~12, no.~3, pp. 2249--2260, May 2021.

\bibitem{tsaousoglou2022market}
G.~Tsaousoglou, J.~S. Giraldo, and N.~G. Paterakis, ``Market mechanisms for
  local electricity markets: A review of models, solution concepts and
  algorithmic techniques,'' \emph{Renewable and Sustainable Energy Reviews},
  vol. 156, p. 111890, Mar. 2022.

\bibitem{eur_com_2019}
A.~Longo, S.~Giaccaria, T.~Bouman, and T.~Efthimiadis, \emph{Societal
  appreciation of energy security. Volume 1, Value of lost load – households
  (EE, NL and PT)}.\hskip 1em plus 0.5em minus 0.4em\relax Publications Office
  of the European Union, 2018.

\bibitem{Beleiu2018}
A.~Sharma, B.~S. Rajpurohit, and S.~Singh, ``A review on economics of power
  quality: Impact, assessment and mitigation,'' \emph{Renewable and Sustainable
  Energy Reviews}, vol.~88, pp. 363--372, May 2018.

\bibitem{gadea2018market}
A.~Gadea, M.~Marinelli, and A.~Zecchino, ``A market framework for enabling
  electric vehicles flexibility procurement at the distribution level
  considering grid constraints,'' in \emph{2018 Power Systems Computation
  Conference (PSCC)}, {D}ublin, Ireland, 11-15 June 2018.

\bibitem{Sarker2017}
M.~R. Sarker, D.~J. Olsen, and M.~A. Ortega-Vazquez, ``Co-optimization of
  distribution transformer aging and energy arbitrage using electric
  vehicles,'' \emph{IEEE Transactions on Smart Grid}, vol.~8, no.~6, pp.
  2712--2722, Nov. 2017.

\bibitem{Sun2021}
X.~Sun and J.~Qiu, ``Hierarchical voltage control strategy in distribution
  networks considering customized charging navigation of electric vehicles,''
  \emph{IEEE Transactions on Smart Grid}, vol.~12, no.~6, pp. 4752--4764, Nov.
  2021.

\bibitem{Yanchong2022}
Y.~Zheng, Z.~Shao, X.~Lei, Y.~Shi, and L.~Jian, ``The economic analysis of
  electric vehicle aggregators participating in energy and regulation markets
  considering battery degradation,'' \emph{Journal of Energy Storage}, vol.~45,
  p. 103770, Jan. 2022.

\bibitem{FREITASGOMES2021}
I.~S. {Freitas Gomes}, Y.~Perez, and E.~Suomalainen, ``Rate design with
  distributed energy resources and electric vehicles: A {C}alifornian case
  study,'' \emph{Energy Economics}, vol. 102, p. 105501, Oct. 2021.

\bibitem{Qingyun2022}
Q.~Nie, L.~Zhang, Z.~Tong, G.~Dai, and J.~Chai, ``Cost compensation method for
  {PEVs} participating in dynamic economic dispatch based on carbon trading
  mechanism,'' \emph{Energy}, vol. 239, p. 121704, Jan. 2022.

\bibitem{brinkel2020should}
N.~Brinkel, W.~Schram, T.~AlSkaif, I.~Lampropoulos, and W.~Van~Sark, ``Should
  we reinforce the grid? {C}ost and emission optimization of electric vehicle
  charging under different transformer limits,'' \emph{Applied Energy}, vol.
  276, p. 115285, Oct. 2020.

\bibitem{Ricardo2022}
R.~A. Daziano, ``Willingness to delay charging of electric vehicles,''
  \emph{Research in Transportation Economics}, p. 101177, Jan. 2022.

\bibitem{Charalampos_2021}
C.~Ziras, T.~Sousa, and P.~Pinson, ``What do prosumer marginal utility
  functions look like? {D}erivation and analysis,'' \emph{IEEE Transactions on
  Power Systems}, vol.~36, no.~5, pp. 4322--4330, Sep. 2021.

\bibitem{Duque_2021}
E.~M. Salazar~Duque, P.~P. Vergara, P.~H. Nguyen, A.~van~der Molen, and J.~G.
  Slootweg, ``Conditional multivariate elliptical copulas to model residential
  load profiles from smart meter data,'' \emph{IEEE Transactions on Smart
  Grid}, vol.~12, no.~5, pp. 4280--4294, Sep. 2021.

\bibitem{Ming_2018}
H.~Ming, A.~A. Thatte, and L.~Xie, ``Revenue inadequacy with demand response
  providers: A critical appraisal,'' \emph{IEEE Transactions on Smart Grid},
  vol.~10, no.~3, pp. 3282--3291, May 2019.

\bibitem{croxton2003comparison}
K.~L. Croxton, B.~Gendron, and T.~L. Magnanti, ``A comparison of mixed-integer
  programming models for nonconvex piecewise linear cost minimization
  problems,'' \emph{Management Science}, vol.~49, no.~9, pp. 1268--1273, Sep.
  2003.

\bibitem{einolander2022multivariate}
J.~Einolander and R.~Lahdelma, ``Multivariate copula procedure for electric
  vehicle charging event simulation,'' \emph{Energy}, vol. 238, p. 121718, Jan.
  2022.

\end{thebibliography}

\end{document}